\documentclass[submission,copyright,creativecommons]{eptcs}
\providecommand{\event}{Applied Category Theory 2020} % Name of the event you are submitting to
\usepackage{graphicx}
\usepackage{prftree, proof, stmaryrd, amssymb, amsthm}
\usepackage{mathtools, float, amstext}
\usepackage{hyperref, enumerate, array}
\usepackage[all,cmtip]{xy}
\usepackage{doi}

\newcolumntype{L}{>{$}l<{$}}

\DeclareMathAlphabet{\mathcal}{OMS}{cmsy}{m}{n}

\newcommand{\blstar}{\mathbf{!L^*}}

\newcommand{\id}{\mathrm{id}}
\newcommand{\ev}{\mathrm{ev}}
\newcommand{\Typ}{\mathrm{Typ}}

\newcommand{\C}{\mathcal{C}}
\newcommand{\F}{\mathcal{F}}

\newcommand{\R}{\mathbb{R}}
\newcommand{\fdVect}{\mathbf{FdVect}}
\newcommand{\Vect}{\mathbf{Vect}}

\newcommand{\Alg}{\mathbf{Alg}}
\newcommand{\Aalg}{\mathbf{Aalg}}

\newcommand{\bs}{\backslash}
\renewcommand{\L}{\mathbf{L}}

\newcommand{\ov}[1]{\overrightarrow{#1}}

\renewcommand{\epsilon}{\varepsilon}
\renewcommand{\phi}{\varphi}

\newtheorem{defn}{Definition}
\newtheorem{remark}{Remark}
\newtheorem{thm}{Theorem}

\newcommand{\semantics}[1]{\llbracket #1 \rrbracket } 

\begin{document}
\title{Categorical Vector Space Semantics for \\ Lambek Calculus with a Relevant Modality \\ (Extended Abstract)}
\def\titlerunning{Vector Semantics for $\blstar$}

\author{Lachlan McPheat\qquad \qquad Mehrnoosh Sadrzadeh
\institute{University College London \\London, UK}
\email{\{m.sadrzadeh,l.mcpheat\}@ucl.ac.uk}
\and
Hadi Wazni
\institute{Queen Mary University London \\ London, UK}
\email{hadi.wazni@hotmail.com}
\and
Gijs Wijnholds
\institute{Utrecht University \\ Utrecht, NL}
\email{g.j.wijnholds@uu.nl}
}
\def\authorrunning{L. McPheat, M. Sadrzadeh, H. Wazni \& G. Wijnholds}
% First names are abbreviated in the running head.
% If there are more than two authors, 'et al.' is used.
%

\maketitle              % typeset the header of the contribution
\begin{abstract}
We develop a categorical compositional distributional semantics for  Lambek Calculus  with a Relevant Modality, $\blstar$,  which has a limited version of the contraction and permutation rules. The categorical part of the semantics is a monoidal biclosed category with a coalgebra modality as defined on Differential Categories. We instantiate this category to finite dimensional vector spaces and linear maps via  ``quantisation" functors and work with three concrete interpretations of the  coalgebra modality. We apply the model to construct categorical and concrete semantic interpretations  for the motivating example of $\blstar$: the derivation of a phrase with a parasitic gap. The  effectiveness of the concrete interpretations  are evaluated via a  disambiguation task,  on an extension of a sentence disambiguation dataset to parasitic gap phrases, using BERT, Word2Vec, and FastText vectors and Relational tensors.   
 %150-200 words

%\keywords{Vector Semantics  \and Differential Category \and Relevant Modality \and Parasitic Gaps \and Linguistic Data \and Disambiguation.}
\end{abstract}

\section{Introduction}
Distributional Semantics of natural language are semantics which model the \textit{Distributional Hypothesis} due to Firth \cite{Firth1957} and Harris \cite{Harris1954} which assumes \textit{a word is characterized by the company it keeps}. Research in Natural Language Processing (NLP) has turned to Vector Space Models (VSMs) of natural language to accurately model the distributional hypothesis. % see \cite{NLP overview} %TODO ADD REFERENCE for an overview.
Such models date as far back as to Rubinstein and Goodenough's co-occurence matrices \cite{RubinsteinGoodenough} in 1965, until today's neural machine learning methods, leading to embeddings, such as Word2Vec \cite{Word2Vec}, GloVe \cite{pennington2014glove}, FastText \cite{FastText} or BERT \cite{BERT} to name a few. VSMs were used even earlier by Salton \cite{Salton64} for information retrieval. %TODO wording 
These models have plenty of applications, for instance thesaurus extraction tasks \cite{Curran2003,Grefenstette1994}, automated essay marking \cite{Landauer1997} and semantically guided information retrieval \cite{Manning2008}.
However, they lack grammatical compositionality, thus making it difficult to sensibly reason about the semantics of portions of language larger than words, such as phrases and sentences. 

Somewhat orthogonally, Type Logical Grammars (TLGs) form highly compositional models of language by accurately modelling grammar, however they lack distributionality, in that such models do not accurately describe the distributional semantics of a word, only its grammatical role.  %TODO Wording
Distributional Compositional Categorical Semantics (DisCoCat)\cite{Coeckeetal2010} combines these two approaches using category theoretic methods, originally developed to model Quantum protocols. DisCoCat has proven its efficacy empirically \cite{GrefenSadrEMNLP, Grefenstette2015, Sadrzadeh2018, wijnholds-sadrzadeh-2019-evaluating, kartsaklis-sadrzadeh-2013-prior,milajevs-etal-2014-evaluating} and has the added utility of being a modular framework which is open to additions and extensions. % example the extension using Frobenius Algebras to model relative pronouns \cite{Sadretal2013Frob,Sadrzadeh2016}.

DisCoCat is a categorical semantics of a formal system which models natural language syntax, known as Lambek Calculus\footnote{There is a parallel pregroup syntax which gives you the same semantics, as discussed in \cite{Coeckeetal2013}}, denoted by $\L$. The work in \cite{Kanovich2016} extends Lambek calculus with a relevant modality, and denotes the resulting logic by $\blstar$. As an example application domain, they use the new logic to formalise the grammatical structure of the parasitic gap phenomena in natural language. 

In this paper, we first form a sound categorical semantics of $\blstar$, which we call $\C(\blstar)$. This boils down to interpreting the logical contraction of $\blstar$ using comonads known as \textit{coalgebra modalities} defined in \cite{St2006}. 
In order to facilitate the categorical computations, we use  the  clasp-string calculus of \cite{BaezStay2011}, developed for depicting the computations of a  monoidal biclosed category. To this monoidal diagrammatic diagrammatic calculus, we add the necessary new constructions  for  the coalgebra modality and its operations.
Next, we define three candidate coalgebra modalities on the category of finite dimensional real vector spaces in order to form a sound VSM of $\blstar$ in terms of structure-preserving functors $\C(\blstar) \to \fdVect_\R$. 
We also briefly introduce a prospective diagrammatic semantics of $\C(\blstar)$ to help visualise our derivations. 
We conclude this paper with an experiment to test the accuracy of the different coalgebra modalitites on $\fdVect_\R$. The experiment is performed using different neural word embeddings and  on a disambiguation task over an extended version of  dataset of [16] from transitive sentences to phrases with parasitic gaps. 

This paper is an extended abstract of the full arXiv paper \cite{mcpheat2020categorical}.

\section{$\blstar$: Lambek Calculus with a Relevant Modality}
Following \cite{Kanovich2016}, we assume that the  formulae, or types,  of  Lambek calculus  with a Relevant Modality $\blstar$ are generated by a  set of atomic types $\mathrm{At}$,  a unary connective $!$, three binary connectives, $\bs$, $/$ and $,$  via   the following Backus-Naur Form (BNF). 
	\[
		 \phi ::= \phi \in At \mid \emptyset \mid (\phi, \phi) \mid (\phi /\phi) \mid (\phi\bs \phi) \mid !\phi ,
	\]
We refer to the types of $\blstar$  by $\Typ_{\blstar}$	; here,  $\emptyset$ denotes the empty type. An element of $\Typ_{\blstar}$	  is   either atomic, made up of a modal type, or  two types joined by a comma or a slash. We will use uppercase roman letters to denote arbitrary types of $\blstar$, and uppercase Greek letters to denote a  set of types, for example, $\Gamma = \{A_1, A_2, \ldots, A_n\} = A_1, A_2, \ldots, A_n$. It is assumed that $,$ is associative, allowing us to omit brackets in expressions like $A_1, A_2, \ldots, A_n$.
	
	A \textbf{sequent} of $\blstar$  is a pair of an  ordered set of types and a type, denoted by $\Gamma \vdash A$. 
	The derivations of $\blstar$ are generated by the  set of  axioms and rules presented in  table \ref{BLrules}.  The logic $\blstar$ extends  Lambek Calculus  $\L$ by endowing it with a modality denoted by $!$, inspired by the $!$ modality of  Linear Logic, to enable the structure rule of contraction in a controlled way, although here it is  introduced on a non-symmetric monoidal category but is introduced with an extra structure allowing the $!$-ed types to commute over other types. So what $\blstar$ adds to $\L$ is the $(!L), (!R)$ rules, the $(\mathrm{perm})$ rules, and the $(\mathrm{contr})$ rule.

	\begin{table}[H]
	\centering
	\scalebox{0.95}{\begin{tabular}{rlrl}
             
             \prftree{A \vdash A} 
             
             & &&
             
             \\
             \\
             
             \prftree[r]{$\scriptstyle{(/ L)}$}
             {\Gamma \vdash A}
             {\Delta_1, B, \Delta_2 \vdash C}
             {\Delta_1, B / A, \Gamma, \Delta_2 \vdash C}
             
             &
             
             \prftree[r]{$\scriptstyle{(/ R)}$}
             {\Gamma, A \vdash B}
             {\Gamma \vdash B/A}
             
           &
             
             \prftree[r]{$\scriptstyle{(\bs L)}$}
             {\Gamma \vdash A}
             {\Delta_1, B, \Delta_2 \vdash C}
             {\Delta_1, \Gamma, A\bs B, \Delta_2 \vdash C}
             
             &
             
             \prftree[r]{$\scriptstyle{(\bs R)}$}
             {A, \Gamma \vdash B}
             {\Gamma \vdash A \bs B}

            \\
            \\
             \prftree[r]{$\scriptstyle{(!L)}$}
                {\Gamma_1, A, \Gamma_2 \vdash C}
                {\Gamma_1, !A, \Gamma_2\vdash C}
                
                &
                \prftree[r]{$\scriptstyle{(!R)}$}
                {!A_1,\ldots, !A_n \vdash B}
                {!A_1,\ldots, !A_n \vdash !B}
                
               &
                
                \prftree[r]{$\scriptstyle{(\mathrm{perm}_1)}$}{\Delta_1, !A,\Gamma, \Delta_2\vdash C}{\Delta_1, \Gamma, !A, \Delta_2 \vdash C}
                
                &
                \prftree[r]{$\scriptstyle{(\mathrm{perm}_2)}$}{\Delta_1, \Gamma, !A, \Delta_2 \vdash C}{\Delta_1, !A,\Gamma, \Delta_2\vdash C}
                
                \\ 
                \\
                
                &&
                
                \prftree[r]{$\scriptstyle{(\mathrm{contr})}$}
                {\Delta_1, !A, !A, \Delta_2 \vdash C}
                {\Delta_1, !A, \Delta_2 \vdash C}
                &
\\
&&&
	\end{tabular}}
	\caption{Rules of $\blstar$.}\label{BLrules}
	\end{table}

\section{Categorical Semantics for $\blstar$}
\label{sec:catsem}

	We  associate $\blstar$ with a category $\C(\blstar)$,  with $\Typ_{\blstar}$ as objects, and derivable sequents  of $\blstar$ as morphisms whose domains are the formulae on the left of the turnstile and codomains the formulae on the right. The category $\C(\blstar)$ is  monoidal biclosed\footnote{We follow the convention that products are not symmetric unless stated, hence a monoidal product is not symmetric unless  referred to by  `symmetric monoidal'.}. The connectives $,$ and $\bs, /$ in $\blstar$ are associated  with the monoidal  structure on $\C(\blstar)$, where $,$ is the monoidal product, with the empty type as its unit and  $\bs, /$ are associated with the two internal hom functors with respect to $,$, as presented in  \cite{Selinger2010}.
 	The connective $!$ of $\blstar$ is a \emph{coalgebra} modality, as defined for \emph{Differential Categories} in \cite{St2006}, with the difference that our underlying category is not necessarily symmetric monoidal, but we ask for a restricted symmetry  with regards to $!$ and that $!$ be a lax monoial functor. In Differential Categories $!$  does not necessarily have a monoidal property, i.e.  it is not a  strict, lax, or strong monoidal functor, but there are examples of Differential Categories  where strong monoidality holds. We elaborate on these notions via the following definition. 
\begin{defn}
\label{def:cblstar}
	The category $\,\C(\blstar)$ has types of $\blstar$, i.e. elements of $\,\Typ_\L$,  as objects, derivable sequents of $\blstar$ as morphisms, together with the following structures:
\begin{itemize}
	\item A \textbf{monoidal product} $\otimes \colon \C(\blstar) \times \C(\blstar) \to \C(\blstar)$,  with a unit $I$.
	\item \textbf{Internal hom-functors} $\Rightarrow : \C(\blstar)^{\mathrm{op}} \times \C(\blstar) \to \C(\blstar)$,  $\Leftarrow : \C(\blstar) \times \C(\blstar)^{\mathrm{op}} \to \C(\blstar)$ such that: 
	\begin{enumerate}[i.]
	\item For objects $A,B \in \C(\blstar)$, we have objects $(A \Rightarrow B), (A \Leftarrow B) \in \C(\blstar)$ and a pair of  morphisms, called \emph{right and left evaluation}, given below:
	\[
	\ev_{A, (A \Rightarrow B)}^r \colon A \otimes (A \Rightarrow B) \longrightarrow A, 
	\qquad
	\ev_{(A \Leftarrow B), B}^l \colon (A \Leftarrow B) \otimes B \longrightarrow A
	\]
	\item For morphisms $f \colon A \otimes C \longrightarrow B, g \colon C \otimes B \longrightarrow A$, we have unique \emph{right and left curried} morphisms, given below:
	\[
	\Lambda^l (f) \colon C \longrightarrow (A \Rightarrow B), \qquad
	\Lambda^r (g) \colon C \longrightarrow (A \Leftarrow B)
	\]
	\item The following hold
	\[
	\ev^l_{A, B} \circ (\id_A \otimes \Lambda^l(f)) = f, \qquad
	\ev^r_{A, B} \circ (\Lambda^r(g) \otimes \id_B) = g
	\]
	\end{enumerate}
	\item A {\bf  coalgebra modality} $!$  on $\C(\blstar)$. That is, a lax monoidal comonad $(!,\delta, \epsilon)$ such that:
	
	\begin{itemize}
	\item[] For every object $A\in \C(\blstar)$, the object $!A$ has a comonoid structure $(!A, \Delta_A, e_A)$ in $\C(\blstar)$. Where the comultiplication $\Delta_A: !A \to !A \otimes !A$, and the counit $e_A: !A \to I$ satisfy the usual comonoid equations. Further, we require  $\delta_A : !A \to !!A$ to be a morphism of comonoids  \cite{St2006}. \footnote{Strictly speaking, this definition applies to symmetric monoidal categories, however we may abuse notation without worrying, as we have symmetry in the image of $!$ coming from the restricted symmetries $\sigma^l, \sigma^r$.}
	\end{itemize}
	
	\item \textbf{Restricted symmetry} over the coalgebra modality, that is, natural isomorphisms $\sigma^r : 1_{\C(\blstar)} \otimes \, ! \to ! \otimes 1_{\C(\blstar)}$ 
	and $\sigma^l : ! \otimes 1_{\C(\blstar)} \to 1_{\C(\blstar)} \otimes \, !$.% That is, for each $A,B\in \C(\blstar)$ we have
	\[
		\sigma^r_{A,B} : A\,  \otimes\,  !B \longmapsto\,  !B\,  \otimes A, \qquad
		\sigma^l_{A,B} : !A\,  \otimes B \longmapsto\, B\,  \otimes\,  !A.
	\]
\end{itemize}
\end{defn}

\noindent
We  now define a categorical semantics for $\blstar$ as the map $\semantics{\ } \colon \blstar \to \C(\blstar)$ and prove that it is sound. 

\begin{defn}
\label{def:sem}
The  \emph{semantics} of   formulae and sequents of  \ $\blstar$ is the image of  the interpretation map $\semantics{\ } \colon \blstar \to \C(\blstar)$. To elements $\phi$ in  $\Typ_\L$,  this map assigns objects $C_\phi$ of $\C(\blstar)$, as defined below:
\[\begin{array}{cccccc}
  \semantics{\emptyset} &:=& C_{\emptyset} =  I &\qquad \qquad \semantics{\phi} &:=& C_{\phi}\\
 \semantics{(\phi, \phi)} &:=& C_{\phi} \otimes C_{\otimes} & \qquad \qquad \semantics{!\phi} &:=& ! C_\phi \\
 \semantics{(\phi / \phi)} &:=& (C_{\phi} \Leftarrow  C_{\phi}) & \qquad \qquad \semantics{(\phi \bs \phi)} &:=&(C_{\phi} \Rightarrow C_{\phi})
\end{array}\]
To the sequents $\Gamma \vdash A$ of $\blstar$, for $\Gamma = \{A_1, A_2, \cdots A_n\}$ where $A_i, A \in \Typ_\L$, it assigns  morphism of  $\C(\blstar)$ as follows $\semantics{\Gamma \vdash A} :=  C_\Gamma \longrightarrow C_A$, for $C_\Gamma = \semantics{A_1} \otimes \semantics{A_2} \otimes \cdots \otimes \semantics{A_n}$.
\end{defn}

\noindent
Since sequents are not labelled, we have no obvious name for the linear map $\semantics{\Gamma \vdash A}$, so we will label such morphisms by lower case roman letters as needed.

\begin{defn}
\label{def:model}
A \textbf{categorical model}  for  $\blstar$, or a  \textbf{$\blstar$-model},  is a pair $(\C, \semantics{ \ }_\C)$, where $\C$ is a monoidal biclosed category with a coalgebra modality and restricted symmetry, and $\semantics{ \ }_\C $ is a mapping $\Typ_{\blstar} \to \C$ factoring through $\semantics{ \ } : \Typ_\L \to \C(\blstar)$. 

%A \emph{categorical model} for  $\blstar$ is the tuple $(\C(\blstar), \semantics{\ })$, for $\C(\blstar)$ as  in  Definition \ref{def:cblstar} and $\semantics{\ }$ as in Definition \ref{def:sem}.  
\end{defn}

%We  go through the rules of $\blstar$ and show that they are sound in the categorical model. As we have equipped $\C(\blstar)$ with adequate structure, this follows easily.  First we define what soundness means in this setting. 

\begin{defn}
\label{def:truth}
A  sequent $\Gamma \vdash A$ of $\blstar$ is sound in $(\C(\blstar), \semantics{\ })$, iff $C_\Gamma \longrightarrow C_A$ is a morphism of $\C(\blstar)$.  A rule $\frac{\Gamma \vdash A}{\Delta \vdash B}$ of $\blstar$ is  \emph{sound} in  $(\C(\blstar), \semantics{\ })$ iff whenever  $C_\Gamma \longrightarrow C_A$ is sound then  so is $C_\Delta \longrightarrow C_B$. We say $\blstar$ is sound with regards to $(\C(\blstar), \semantics{\ })$ iff  all of its rule are.
\end{defn}

\begin{thm}
\label{thm:sound}
$\blstar$ is sound with regards to $(\C(\blstar), \semantics{ \ })$. 
\end{thm}

\begin{proof}
See  full paper \cite{mcpheat2020categorical}.
\end{proof}

\section{Vector Space Semantics for ${\cal C}(\blstar)$}

	Following \cite{Coeckeetal2013}, we  develop vector space semantics for  $\blstar$, via a \emph{quantisation} functor  to the category of finite dimensional vector spaces and linear maps $F: \C(\blstar) \to \fdVect_\R$. This functor  interprets objects as finite dimensional vector spaces, and derivations as linear maps. Quantisation is the term first introduced by Atiyah in Topological Quantum Field Theory, as a functor from the category of manifolds and cobordisms to the category of vector spaces and linear maps. Since the cobordism category is monoidal, quantisation was later  generalised to refer to a functor that `quantises' any   category in  $\fdVect_\R$.
Since $\C(\blstar)$ is free, there is a unique functor $\C(\blstar) \to (\fdVect_\R, !)$ for any choice of $!$ such that $(\fdVect_\R, !)$ is a $\blstar$-model. In definition \ref{def:quantF} we introduce the necessary nomenclature to define quantisations in full.
	
\begin{defn}
\label{def:quantF}
A \textbf{quantisation} is a functor $F : \C(\blstar) \to (\fdVect_\R,!)$, defined on the objects of $\C(\blstar)$ using the structure of the formulae  of $\blstar$, as follows: 
	\begin{align*}
		F(C_{\emptyset}) &:= \R& \qquad  F(C_\phi) &:= V_\phi\\
            	F(C_{\phi \otimes \phi}) & :=  V_\phi \otimes V_\phi &  \qquad  F (C_{!\phi}) &:=!V_\phi \\
            	F (C_{\phi / \phi}) &:=(V_\phi \Leftarrow V_\phi) & \qquad 	F(C_{\phi \bs \phi}) &:=  (V_\phi \Rightarrow V_\phi)
	\end{align*}
Here,  $V_\phi$ is the vector space in which vectors of words with an atomic type live  and the other vector spaces are obtained from it by induction on the structure of the formulae they correspond to. Morphisms  of $\C(\blstar)$ are of the form $C_\Gamma \longrightarrow C_A$, associated with sequents $\Gamma \vdash A$ of $\blstar$, for $\Gamma = \{A_1, A_2, \cdots, A_n\}$.  The quantisation functor is defined on these morphisms as follows:
\[
F(C_{\Gamma} \longrightarrow C_A) := F(C_\Gamma) \longrightarrow F(C_A) = V_{A_1} \otimes V_{A_2} \otimes \cdots \otimes V_{A_n} \longrightarrow V_A
\]

\end{defn}

	Note that the monoidal product in $\fdVect_\R$ is symmetric, so there is formally no need to distinguish between $(\semantics{A} \Rightarrow \semantics{B})$ and $(\semantics{B} \Leftarrow \semantics{A})$. However it may be practical to do so when calculating things by hand, for example when retracing derivations  in the semantics. 
	We should also make clear that the freeness of $\C(\blstar)$ makes $F$ a \emph{strict} monoidal closed functor, meaning that $F(C_A \otimes C_B) = FC_A \otimes FC_B$, or rather, $V_{(A\otimes B)} = (V_A \otimes V_B)$, and similarly, $V_{(A\Rightarrow B)} = (V_A \Rightarrow V_B)$ etc.
	Further, since we are working with finite dimensional vector spaces we know that $V_\phi^\bot \cong V_\phi$, thus our internal homs have an even simpler structure, which we exploit when computing, which is $V_\phi \Rightarrow V_\phi \cong V_\phi \otimes V_\phi$.
	
\section{Concrete Constructions}
\label{sec:concrete}

In this section we present three different coalgebra modalities on $\fdVect_\R$ defined over two different underlying comonads, treated in individual subsections. Defining these modalities lets us reason about sound vector space semantics of $\C(\blstar)$ in terms of $!$-preserving monoidal biclosed functors $\C(\blstar) \to \fdVect_\R$. 

We point out here that we do not aim for \textit{complete} model in that we do not require the tensor of our vector space semantics to be non-symmetric. This is common practice in the DisCoCat line of research and also in the  standard set theoretic semantics of Lambek calculus \cite{VanBenthem1988}. Consider the English sentence "John likes Mary" and the Farsi sentence "John Mary-ra Doost-darad(likes)". These two sentences have the same semantics, but different word orders, thus exemplifying the lack of syntax within semantics.

\subsection{$!$ as the Dual of a Free Algebra Functor}%TODO work out better title "! as dual of free algebra functor" I think... 
\label{sec:FSpace}
Following \cite{Blute1994} we interpret $!$ using the Fermionic Fock space functor $\F : \fdVect_\R \to \Alg_\R$. In order to define $\F$ we first introduce the simpler free algebra construction, typically studied in the theory of representations of Lie algebras \cite{Humphreys}. It turns out that $\F$ is itself a free functor, giving us a comonad structure on $U\F$ upon dualising \cite{Blute1994}. The choice of the symbol $\F$ comes from ``Fermionic Fock space" (as opposed to ``Bosonic"), and is also known as the exterior algebra functor, or the Grassmannian algebra functor \cite{Humphreys}.

\begin{defn}
\label{def:T}
The {\bf free algebra functor} $T : \Vect_\R \to \Alg_\R$ is defined on objects as:
	\[ V  \longmapsto  \bigoplus_{l\geq 0} V^{\otimes n} = \R \oplus V \oplus (V \otimes V) \oplus (V\otimes V \otimes V) \otimes \cdots \]
and for morphisms $f : V \to W$, we get the algebra homomorphism $T(f) : T(V) \to T(W)$ defined layer-wise as
	\[ T(f) (v_1 \otimes v_2 \otimes \cdots \otimes v_n) := f(v_1) \otimes f(v_2) \otimes \cdots \otimes f(v_n).\]
\end{defn}

$T$ is free in the sense that it is left adjoint to the forgetful functor $U : \Alg_\R \to \Vect_\R$, thus giving us a monad $UT$ on $\Vect_\R$ with a  monoidal algebra modality structure, i.e. the dual of what we are looking for.  However note that even when restricting $T$ to finite dimensional vector spaces $V \in \fdVect_\R$ the resulting  $UT(V)$ and $UT(V^\bot)^\bot$  are infinite-dimensional. The necessity of working in $\fdVect_\R$ motivates us to use $\F$, defined below, rather than $T$.

\begin{defn}
\label{def:F}
The {\bf Fermionic Fock space functor} $\F : \Vect_\R \to \Alg_\R$\footnote{One may wish to think of $\F$ as having codomain $\Aalg_\R$, the category of antisymmetric $\R$-algebras, which is itself a subcategory $\Aalg_\R \hookrightarrow \Alg_\R$.} is defined on objects as 
	\[V \mapsto \bigoplus_{n \geq 0} V^{\wedge n} = \R \oplus V \oplus (V \wedge V) \oplus (V\wedge V \wedge V) \otimes \cdots\]
where $V^{\wedge n}$ is the coequaliser of the family of maps $(-\tau_\sigma)_{\sigma \in S_n}$, defined as $-\tau_\sigma : V^{\otimes n} \to V^{\otimes n}$ and given as follows:
	\[(-\tau_{\sigma}) ( v_1 \otimes \cdots \otimes v_n) := \mathrm{sgn}(\sigma) (v_{\sigma(1)} \otimes v_{\sigma(2)} \otimes \cdots \otimes v_{\sigma(n)}).\]	
	$\F$ applied to linear maps gives an analogous algebra homomorphism as in \ref{def:T}. 	
\end{defn}

Concretely, one may define $V^{\wedge n }$ to be the $n$-fold tensor product of $V$ where we quotient by the layer-wise equivalence relations $v_1 \otimes v_2 \otimes \cdots \otimes v_n \sim \mathrm{sgn}(\sigma)(v_{\sigma(1)} \otimes v_{\sigma(2)} \otimes \cdots \otimes v_{\sigma(n)})$ for $n = 0, 1, 2 \ldots$ and denoting the equivalence class of a vector $v_1 \otimes v_2 \otimes \cdots \otimes v_n$ by $v_1 \wedge v_2 \wedge \cdots \wedge v_n$. 

Note that simple tensors in $V^{\wedge n}$ with repeated vectors are zero. That is, if $v_i = v_j$ for some $1 \leq i,j\leq n$ and $i\neq j$ in the above, the permutation $(ij) \in S_n$ has odd sign, and so $v_1 \wedge v_2 \wedge \cdots \wedge v_n = 0$, since $v_1 \wedge v_2 \wedge \cdots \wedge v_n  = \mathrm{sgn}(ij) (v_1 \wedge v_2 \wedge \cdots \wedge v_n) = - (v_1 \wedge v_2 \wedge \cdots \wedge v_n)$.

\begin{remark}\label{rem:fDim}
Given a finite dimensional vector space $V$, the antisymmetric algebra $\F(V)$ is also finite dimensional. This follows immediately from the note in definition \ref{def:F}, as basis vectors in layers of $\F(V)$ above the $\dim(V)$-th are forced to repeat entries.
\end{remark}

Remark \ref{rem:fDim} shows that restricting $\F$ to finite dimensional vector spaces turns $U\F$ into an endofunctor on $\fdVect_\R$. We note that $\F$ is the free antisymmetric algebra functor \cite{Blute1994} and conclude that $U\F$ is a monad $(U\F, \mu, \eta)$ on $\fdVect_\R$. 

Given $\F$, there are two ways to obtain a comonad structure $(U\F(V), \Delta_V, e_V)$, thus define a coalgebra modality $(U\F, \delta, \epsilon)$ on $\fdVect_\R$, as desired. One is referred to by \emph{Cogebra} construction and is given below, for a basis $\{e_i\}_i$ of  $V$, and thus a basis $\{1, e_{i_1}, e_{i_2} \wedge e_{i_3}, e_{i_4} \wedge e_{i_5} \wedge e_{i_6}, \cdots\}_{i_j}$ of $U\F(V)$ as:
\[
\Delta(1, e_{i_1}, e_{i_2} \wedge e_{i_3}, e_{i_4} \wedge e_{i_5} \wedge e_{i_6}, \cdots) =
(1, e_{i_1}, e_{i_2} \wedge e_{i_3}, e_{i_4} \wedge e_{i_5} \wedge e_{i_6}, \cdots)\otimes (1, e_{i_1}, e_{i_2} \wedge e_{i_3}, e_{i_4} \wedge e_{i_5} \wedge e_{i_6}, \cdots)
\]

The map $e_V:U\F(V) \to V$ is given by projection onto the first layer, that is
\[1, e_{i_1}, e_{i_2} \wedge e_{i_3}, e_{i_4} \wedge e_{i_5} \wedge e_{i_6}, \cdots \longmapsto e_{i_1}.\]

Another coalgebra modality arises from dualising the monad $U\F$, and the monoid structure on $\F(V)$, or strictly speaking on $U\F (V)$. %; this is referred to by \emph{cofree} construction $\cdots$. 
Following \cite{HopfMonads}, we dualise $U\F$ to define a comonad structure on $U\F$ as follows. We take the comonad comultiplication to be $\delta_V := \mu_V^\bot : U\F U\F (V)^\bot \to U\F (V)^\bot$, and the comonad counit to be $\epsilon_V := \eta_V^\bot : U\F(V)^\bot \to V^\bot$. To avoid working with dual spaces one may chose to formally consider $!(V) := U\F(V^\bot)^\bot$, as in \cite{Blute1994}, since $U\F(V^\bot)^\bot \cong U\F (V)$ (although this is not strictly necessary, we choose this notation to stay close to its original usage \cite{Blute1994,HopfMonads}). 
Note that a dualising in this manner only makes sense for finite dimensional vector spaces, as  in general, for an arbitrary family of vector spaces $(V_i)_{i\in I}$, we have $(\bigoplus_{i\in I} V_i)^\bot \cong \prod_{i\in I} (V_i^\bot)$. 
Finite dimensionality of a vector space $V$ makes the direct sum in $U\F(V)$ finite, making the right-hand product a direct sum, i.e. for a finite index $I$ we have $(\bigoplus_{i\in I} V_i)^\bot \cong \bigoplus_{i\in I} (V_{i\in I}^\bot)$, meaning that we indeed have $U\F(V)^\bot \cong U\F(V^\bot)$.
This lets us dualise the monoid structure of $U\F(V)$, giving a comonoid structure on $U\F(V)$ hence making $U\F$ into a coalgebra modality. To compute the comultiplication it suffices to transpose the matrix for the multiplication on $U\F(V)$. However, this is in general intractable, as for $V$ an $n$ dimensional space, $U\F(V)$ will have $2^n$ dimensions, and its multiplication will be a $(2^n)^2 \times  2^n$-matrix. We leave working with a dualised comultiplication to another paper, but in the next subsection use this construction to obtain a richer copying than the Cogebra one mentioned above. 

\subsection{$!$ as the Identity Functor}

The above Cogebra construction can be simplified  when one works with free vector spaces, for details of which we refer to the full version of the paper \cite{mcpheat2020categorical}.  The simplified version resembles half of a bialgebra over $\fdVect_\R$, known as \emph{Special Frobenius bialgebras}, which were used in \cite{Sadretal2013Frob,MoortgatWijn2017,Moortgatetal2019} to model relative pronouns in English and Dutch. As argued in \cite{WijnSadr2017}, however, the copying map resulting from this comonoid structure only copies the basis vectors and does not seem adequate for a full copying operation. In fact, a quick computation shows that this $\Delta$ in a sense only half-copies of the input vector.  In order to see this, consider a vector $\ov{v} = \sum_i C_i s_i$, for $s_i \in S$. Extending the comultiplication $\Delta$ linearly provides us with 
\[
\Delta(\ov{v}) = \sum_i C_i \Delta(s_i) = \sum_i C_i (s_i \otimes s_i) = (\sum_i C_i s_i ) \otimes (\sum_i s_i) =\ov{v} \otimes \vec{1},
\]
In the second term, we have lost the $C_i$ weights, in other words we have replaced the second copy with a vector of 1's, denoted by $\vec {1}$. 
 
The above problem can be partially overcome by noting that this $\Delta$ map is just one of a family of copying maps, parametrised by reals, where for any $k\in \R$ we may define the a \emph{Cofree-inspired} comonoid  $(V_\phi, \Delta_k, e)$ over a vector spafce $V_\phi$ with a basis $(v_i)_i$, as: 
		\[
		\Delta_k: V_\phi \to V_\phi \otimes V_\phi  :: v \mapsto (v \otimes \vec{k}) + (\vec{k} \otimes v), \quad 
		e: V_\phi \to \R :: \sum_i C_i v_i \mapsto \sum_i C_i
		\]
Here,  $\ov{v}$ is as before and $\vec{k}$ stands  for an element of $V$ padded with number $k$. In the simplest case, when $k=1$, we obtain  two copies of the weights $\ov{v}$ and also of its basis vectors, as the following calculation demonstrates. Consider a two dimensional vector space and the vector $a e_1 + b e_2$ in it. The 1 vector $\ov{1}$ is the 2-dimensional vector $e_1 + e_2$ in $V$. Suppose $\ov{v}$ and $\vec{1}$ are column vectors, then applying $\Delta$  results in the matrix $2a\,  e_1\otimes e_1 + ab\,  e_1 \otimes e_2 + ab\,  e_2 \otimes e_1 + 2b\,  e_2 \otimes e_2$, where we have two copies of the weights in the diagonal and also the basis vectors have obviously multiplied.

This construction is inspired by the graded algebra construction on vector spaces, whose dual construction is referred to as a \emph{Cofree Coalgebra}. The  Cofree-inspired coalgebra over a vector space defines a coalgebra modality structure on the identity comonad on $\fdVect_\R$, which provides another $\blstar$-model, or rather, another quantization $\C(\blstar) \to \fdVect_\R$.

\section{Clasp Diagrams}
\label{sec:claspDiagrams}

%	\usetikzlibrary{arrows,decorations,backgrounds,positioning,fit}
%	\tikzset{func/.style={shape=rectangle,rounded corners=8,minimum width=2cm,minimum height=.5cm,draw}}
%	\tikzset{claspnode/.style={shape=circle,minimum width=0.25cm,fill=white,draw}}
%

In order to show the semantic computations for the parasitic gap, we introduce a diagrammatic semantics. The derivation of the parasitic gap phrase is involved and its categorical or vector space interpretations require close inspection to read. The diagrammatic notation makes it easier to visualise the steps of the derivation and the final semantic form. In what follows we first introduce notation for the Clasp diagrams, then extend them with extra prospective notation necessary to model the $!$ coalgebra modality.
The basic structure of the $\C(\blstar)$ category, i.e. its objects, morphisms, monoidal product and its internal homs are as in \cite{BaezStay2011}. To these, we add the necessary diagrams for the coalgebra modality, that is the coalgebra comultiplication (copying) $\Delta$, the counit of the comonad $\epsilon$, and the comonad comultiplication $\delta$, found in figure \ref{fig:diagrams}.

\begin{figure}[!t]
%	\tikzfig{bangDiagrams}
	\includegraphics[scale=0.3]{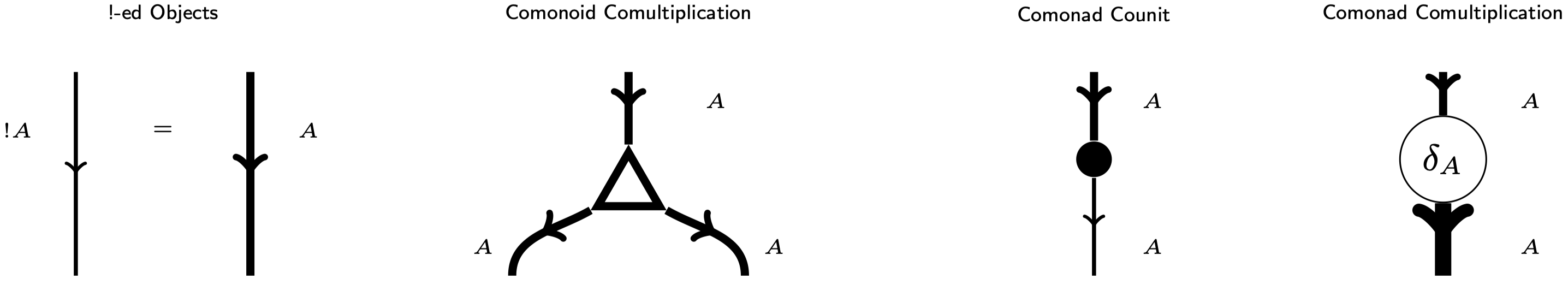}
    \caption{Diagrammatic Structure of $!$ in $\C(\blstar)$}
\label{fig:diagrams}
\end{figure}

\section{Linguistic Examples}
\label{sec:ling}

The motivating example of \cite{Kanovich2016} was the parasitic gap example ``the paper that John signed without reading",  with the following lexicon:
\begin{eqnarray*}
&&	\left \{ (\mbox{The}, NP \bs N), (\mbox{paper}, N), (\mbox{that}, (N\bs N)/(S/!NP)), (\mbox{John}, NP), (\mbox{signed}, (NP\bs S)/NP),   \right. \\
 && \left.(\mbox{without}, ((NP\bs S)\bs(NP \bs S))/NP), (\mbox{reading}, NP / NP)\right \}.
\end{eqnarray*}

The  $\blstar$ derivation  of ``the paper that John signed without reading" is in the full version of the paper \cite{mcpheat2020categorical}.  The categorical semantics of this  derivation is the following linear map.
\begin{eqnarray*}
 &&   (\semantics{NP}\Leftarrow \semantics{N}) \otimes \semantics{N} \otimes ((\semantics{N} \Rightarrow \semantics{N})\Leftarrow (\semantics{S} \Leftarrow \semantics{!NP})) \otimes  \semantics{NP}   \otimes ((\semantics{NP} \Rightarrow \semantics{S})\Leftarrow \semantics{NP})  \otimes \\
 &&  
    (((\semantics{NP} \Rightarrow \semantics{S}) \Rightarrow (\semantics{NP} \Rightarrow \semantics{S}))\Leftarrow \semantics{NP}) \otimes
    (\semantics{NP} \Leftarrow \semantics{NP}) 
    \longrightarrow
    \semantics{NP}
\end{eqnarray*}
defined on elements as follows, for the bracketed subscripts  in Sweedler notation:
\begin{align*}
	the(-) \otimes \ov{paper} \otimes that(-,-) \otimes \ov{John} \otimes signed(-,-) \otimes without( -,-,-) \otimes reading(-) \\
	\mapsto \quad the(that(\ov{paper}, without(\ov{John}, signed(-,-_{(1)}), reading(-_{(2)}))))
\end{align*}

The  diagrammatic interpretation of the $\blstar$-derivation is depicted in figure \ref{fig:pGapEx}
\begin{figure}
	\centering
	\includegraphics[scale=0.35]{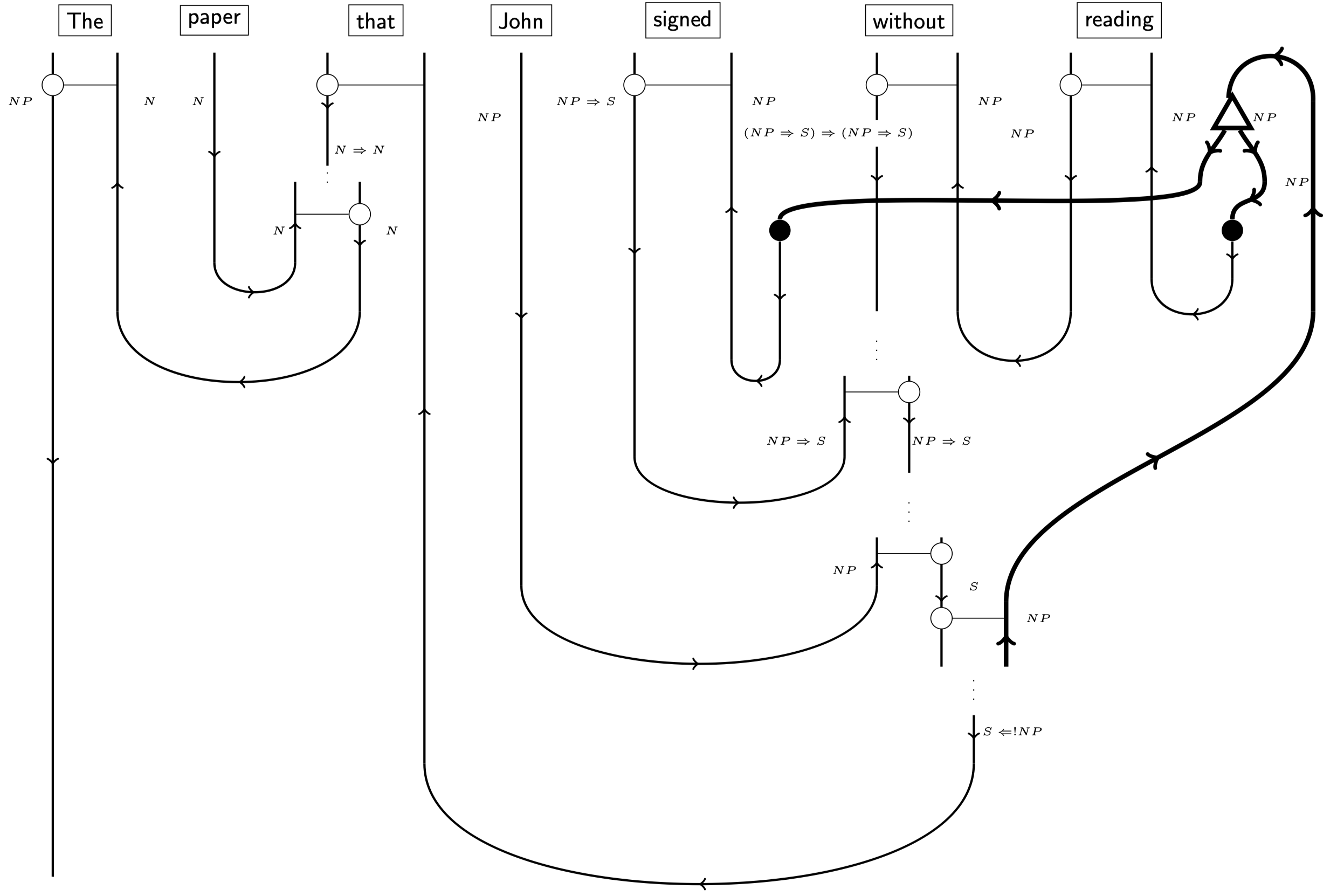}
	\caption{Diagrammatic interpretation of ``The paper that John signed without reading"}\label{fig:pGapEx}
\end{figure}

This is obtained via steps mirroring the steps of the derivation tree of the example,  please see the full version of the paper \cite{mcpheat2020categorical}.
	
\section{Experimental Comparison}\label{sec:experiments}
The reader might have rightly been wondering which one of these interpretations, the Cogebra or the Cofree-inspired coalgebra model, produces the correct semantic representation. 
%Two options are possible here: one is similar  to the line of work done in \cite{HedgesSadr2019,Sadretal2013Frob}, where the formal semantics of the sentence is embedded  in vector spaces and the result is considered to be the vector space version of the correct meaning. 
We implement  the resulting vector   representations  on  large corpora of data  and experiment with a disambiguation task to provide insights.  The disambiguation task was that  originally proposed in \cite{GrefenSadrEMNLP}, but we work  with  the  data set of \cite{KartSadrCoNLL}, which  contains verbs deemed as \emph{genuinely ambiguous} by  \cite{PickeringFrisson}, as those verbs whose meanings are not related to each other. We extended this latter with a second verb and a preposition that provided enough data to turn the dataset from a set of pairs of transitive sentences to a set of pairs of  parasitic gap phrases.  As an example, consider the verb {\bf file}, with meanings {\bf register} and {\bf smooth}. Example entries of the original  dataset and its extension are below; the full dataset is available from \url{https://msadrzadeh.com/datasets/}.

\begin{figure}[!h]
\centering
\begin{tabular}{l}
S: accounts  that the local government  filed\\
 S1: accounts  that the local government  registered\\
 S2: accounts  that the local government smoothed \\
  P: accounts  that the local government  filed after inspecting\\
P1: accounts  that the local government  registered after inspecting\\
P2: accounts  that the local government smoothed after inspecting
 \end{tabular}
\\  \vspace{0.5cm}
 \begin{tabular}{l}

P': nails that the young woman filed after  cutting \\
P'1: nails that the young woman registered after  cutting \\
P'2: nails that the young woman smoothed after  cutting \\
 S': nails that the young woman filed\\
 S'1: nails that the young woman registered\\
 S'2: nails that the young woman smoothed

 \end{tabular}
 \vspace{0.2cm}
\end{figure}
 
 We  follow the same procedure  as in  \cite{KartSadrCoNLL} to disambiguate the phrases with the ambiguous verb:  (1) build vectors for phrases P, P1,  P2, and also P', P'1,  P'2, (2) check whether vector of P is closer to vector of P1 or vector of P2 and whether P' is close to P'2 or P'1. If yes, then we have two correct outputs, (3) compute a mean average  precision (MAP), by counting in how many of the pairs, the vector of the phrase with the ambiguous verb is closer to that of the phrase with its appropriate meaning.

In order to instantiate our categorical model on this task and experiment with the different copying maps, we proceed as follows. We work with the parasitic gap phrases that have the general form: ``A’s the B C’ed Prep D’ing". Here,   C and D are verbs and their vector representations are multilinear maps. C is a bilinear map  that takes A and B as input and D is a linear map that takes A as input. For now, we represent the preposition Prep by the trilinear map $Prep$. The vector representation of the parasitic gap phrase with a proper copying operator is $Prep({C}(\ov{B},\ov{A}),  D(\ov{A}))$, for $C$ and $D$ multilinear maps and $\ov{A}$ and $\ov{B}$, vectors, and where $\ov{A} =  \sum_i C^A_i n_i$. The different types of copying applied to this, provide us with the  following options. 
\begin{eqnarray*}
\mbox{\bf Cogebra copying} &&
(a)\,   Prep  \left( C(\ov{B}, \ov{A}),  D(\sum_i n_i) \right), \quad (b)\,  Prep  \left( C(\ov{B}, \sum_i n_i),  D(\ov{A})\right)\\
\mbox{\bf Cofree-inspired copying} && Prep \left(C(\ov{B}, \ov{A}) + D(\vec{1}), C(\ov{B},\vec{1}) + D(\ov{A})\right)
%\mbox{\bf full copying} && Prep  \left( C(\ov{B}, \ov{A}),  D(A) \right)
\end{eqnarray*}
In the \emph{copy object} model of \cite{KartSadrCoNLL}, these choices become as follows:
\begin{eqnarray*}
\mbox{\bf Cogebra copying} && (a)\, Prep  \left(\ov{A} \odot ({C} \times \ov{B}),  {D} \times \sum_i n_i \right)\\
&&(b)\, Prep  \left((\sum_i n_i) \odot ({C} \times \ov{B}),  ({D} \times \ov{A})\right)\\
\mbox{\bf Cofree-inspired copying} &&Prep \left((\ov{A} \odot ({C} \times \ov{B}))  +  ({D} \times \vec{1}), (\vec{1} \odot ({C} \times \vec{B})) + ({D} \times \vec{A})\right)
%\mbox{\bf full copying} && Prep  \left(\ov{A} \odot ({C} \times \ov{B}),  {D} \times \ov{A} \right)
\end{eqnarray*}

\noindent
For comparison, we also  implemented a  model where a \emph{Full} copying operation $\Delta(\ov{v}) = \ov{v} \otimes \ov{v}$ was used, resulting in a third option  $Prep  \left( C(\ov{B}, \ov{A}),  D(A) \right)$, with the \emph{copy-object} model 
\[
Prep  \left(\ov{A} \odot ({C} \times \ov{B}),  {D} \times \ov{A} \right)
\]
Note that this copying is non-linear and thus cannot be an instance of our $\fdVect_\R$ categorical semantics; we are only including it to study how the other copying models  will do in relation to it.   

\begin{table}[!t]
\caption{Parasitic Gap Phrase Disambiguation Results}
\centering
\begin{tabular}{|c||c||c|c||c|c|}
\hline
 Model &  MAP &  Model &  MAP  & Model & MAP \\
\hline \hline
  BERT & 0. 65 & FT(+) & 0.55  & W2V (+) & 0.46 \\
  \hline
Full & 0.48 & Full &  0.57 & Full & 0.54 \\
  Cofree-inspired & 0.47 & Cofree-inspired  & 0.56 & Cofree-inspired  & 0.54 \\
 Cogebra (a)  & 0.46 &  Cogebra (a) & 0.56 & Cogebra (a) & 0.46  \\
  Cogebra (b) & 0.42 & Cogebra (b)  & 0.37 & Cogebra (b)  & 0.39\\
\hline
\end{tabular}
\centering
\label{tab:exampleres}
\end{table}

The results of experimenting with these models are presented  in table \ref{tab:exampleres}. We experimented with three neural embedding architectures: BERT \cite{BERT}, FastText (FT) \cite{FastText}, and Word2Vec CBOW (W2V) \cite{Word2Vec}. For details of the training, please see the full version of the paper \cite{mcpheat2020categorical}.

Uniformly,   in all the  neural architectures, the Full model provided a better disambiguation than other linear copying models. This better performance was closely followed by the Cofree-inspired model: in BERT, the Full model obtained an MAP of 0.48, and the Cofree-inspired model an MAP of 0.47; in FT, we have 0.57 for Full and 0.56 for Cofree-inspired; and in W2V we have 0.54 for both models. Also uniformly, in all of the neural architectures, the Cogebra (a)   did better than the Cogebra (b). It is not surprising that the Full copying  did better than other two copyings, since this is the model that provides two identical copies of the head noun $A$. This kind of copying can only be obtained via the application of a non-linear $\Delta$. The fact that our linear Cofree-inspired copying  closely followed the Full model, shows that in the absence of  Full copying, we can always use the Cofree-inspired as a reliable approximation.  It was also not surprising  that  the Cofree-inspired model did better than either of   the Cogebra models, as this model uses the sum of the two possibilities, each encoded in one of the  Cogebra (a) or (b). That Cogebra (a) performed better than Cogebra (b),  shows that it is more important to have a full copy of the object for the main verb rather than the secondary verb of a parasitic gap phrase. Using this, we can say that  verb $C$ that got a full copy of its object $A$,  played a more important role in disambiguation, than verb $D$,  which only got a vector of 1's as a copy of $A$.    Again, this is  natural, as the secondary verb only provides subsidiary information.    

The most effective disambiguation  of the new dataset was obtained via the BERT phrase vectors, followed by the {Full} model.  BERT  is a contextual neural network architecture that provides different meanings for words in different contexts, using a large set of tuned parameters   on large corpora of data. There is evidence that BERT's phrase vectors do encode some grammatical information in them. So it  is not surprising that these embeddings  provided the best disambiguation result. In the other neural embeddings: W2V and FT, however, the Full and its Cofree-inspired approximation provided better results. Recall that in these models, phrase embeddings are obtained by  adding the word embeddings, and  addition forgets the grammatical structure. That the type-driven categorical model  outperformed these models   is a very promising  result.

\section{Future Directions}\label{sec:future}
There are plenty of questions that arise from the theory in this paper, concerning alternative syntaxes, coherence, and optimisation.

One avenue we are pursuing is to bound the $!$-modality of $\blstar$. This is desirable from a natural language point of view, as the $!$ of linear logic symbolises \emph{infinite} reuse, however at no point in natural language is this necessary. Thus bounding $!$ by indexing with natural numbers, similar to Bounded Linear Logic \cite{Girard1992} may allow for a more intuitive notion of resource insensitivity closer to that of natural language.

Showing the coherence of the diagrammatic semantics by using the proof nets of Modal Lambek Calculus \cite{moortgat1996multimodal}, developed for clasp-string diagrams in \cite{wijnholds2017coherent} constitutes work in progress.  Proving coherence would allow us to do all our derivations diagrammatically, making the sequent calculus labour superfluous. However, we suspect there are better notations for the diagrammatic semantics perhaps more closely related to the proof nets of linear logic.

Applications of type-logics with limited contraction and permutation to movement phenomena is a line of research initiated in \cite{Morrilletal1990,Barryetal1995} with a recent boost in \cite{MV2016,Morrill2017,Morrill2018}, and also in \cite{WijnSadr2019}. Finding commonalities with these approaches is future work. 

We would also like to see how much we can improve the implementation of the cofree-inspired model in this paper. This involves training better tensors, hopefully by using neural networks methods. 

\section{Acknowledgement}
Part of the motivation behind this work came from the Dialogue and Discourse Challenge project of the Applied Category Theory adjoint school during the week 22–26 July 2019. We would like to  thank the organisers of the  school.  We would also like to thank  Adriana Correia, Alexis Toumi, and Dan Shiebler, for discussions. McPheat acknowledges support from the UKRI EPSRC Doctoral Training Programme scholarship, Sadrzadeh from the Royal Academy of Engineering Industrial Fellowship IF-192058. 

\bibliographystyle{plainurl}
\bibliography{references.bib}

\end{document}